\documentstyle[times,afd-gen,afd-orm,asymetrix,url,eg]{article}

\title{Generating Significant Examples for\Eol
       Conceptual Schema Validation\EOL
       {\normalsize Asymetrix Report 94-3}
}
\author{
   H.A. Proper\\
   Asymetrix Research Laboratory\\
   Department of Computer Science\\
   University of Queensland\\
   Australia 4072\\
   E.Proper@acm.org}
\date{\Version}

   \input{epsf}
   \def\Figurepath{./}
   \def\Scale{0.9}
   \def\epsfsize#./##2{\Scale#./}

\begin{document}
   \maketitle
   {\sc Published as:}
\begin{quote}
  H.A.~(Erik) {Proper}. {Generating Significant Examples for Conceptual Schema Validation}. Technical report, Asymetrix Research Laboratory, University of Queensland, Brisbane, Queensland, Australia, 1994.
\end{quote}

   \begin{abstract}
   This report bases itself on the idea of using concrete examples to verify 
   conceptual schemas, and in particular cardinality constraints.
   When novice ORM modellers model domains, the selection of proper cardinality
   constraints for relationship types is quite often prone to errors.
   In this report we propose a mechanism for the generation of significant 
   examples for selected subschemas. 
   The generated examples are significant in the sense that they illustrate 
   the possible combinations of instances that are allowed with respect to the
   cardinality constraints on the involved relationship types.

   In this report we firstly provide a brief informal discussion of the basic
   idea. 
   Then we present a syntactic mechanism to select the subschema for which 
   example instances are to be generated.
   This is followed by the actual example generation algorithm itself.
   We will also present, as a {\em spin-off}, an algorithm that allows us to 
   detect possible flaws in the conceptual schema by calculating the number of 
   instances that can be used to populate the types in the schema.
\end{abstract}
 
   \section{Introduction}

A key aspect in the conceptual design procedure (\cite{Book:94:Halpin:ORM})
is the use of examples to derive the initial design of conceptual schemas. 
A further use of the examples is the validation of parts of the final 
conceptual schema. 
Example populations of relationship types can be used to validate the correctness of 
the information structure, and even more importantly, for the validation of 
constraints. 
In this report we propose a mechanism to generate significant example 
populations for subsets of conceptual schemas.
This idea has been described informally in \cite{AsyReport:94:Harding:ExView}.
In general, generating a significant example population is hardly possible.
Therefore, we limit ourselves to cardinality constraints and significance
of the examples with respect to single relationship types only.

The sample populations are used to visualise for the users what the effects 
are of adding or removing cardinality constraints on relationship types, and 
to a lesser extent that of changing the information structure itself. 
An informal discussion of this idea can be found in 
\cite{AsyReport:94:Harding:ExView}.
This proposed mechanism will initially be used in DBCreate and is expected to 
migrate to InfoModeler in a later stage; after the idea has been tested by 
the user community.

The aim of DBCreate is to enable (semi!) laymen to design their own database.
The example generation tool fits quite well into that idea.
Note that it is only obvious that these semi laymen are still presumed to have 
some basic knowledge of conceptual (and preferably ORM) modelling. 
No matter how user friendly a CAD/CAM system may become, an architect desiging
a house is still required to have a basic working knowledge of the design of 
houses.

As stated before, in general it is nearly impossible to construct sample 
populations that are truely significant with respect to all constraints 
(\cite{Report:90:Bommel:PredMod}). 
In this report we therefore define a significant population of a relationship 
type to be a population that shows all allowed combinations of instances with
respect to the cardinality constraints defined on that relationship type.
Setting a less limited goal can easily lead to a combinatoric explosion.
The restrictions we made, however, are not an unreasonable limitation in the 
context of our aims. 
The most commonly used (and thus mis-used) constraints are the cardinality
constraints (totality and uniqueness).
For the other constraint types there is of course still the possibility 
of verbalising them in a semi-natural language format.

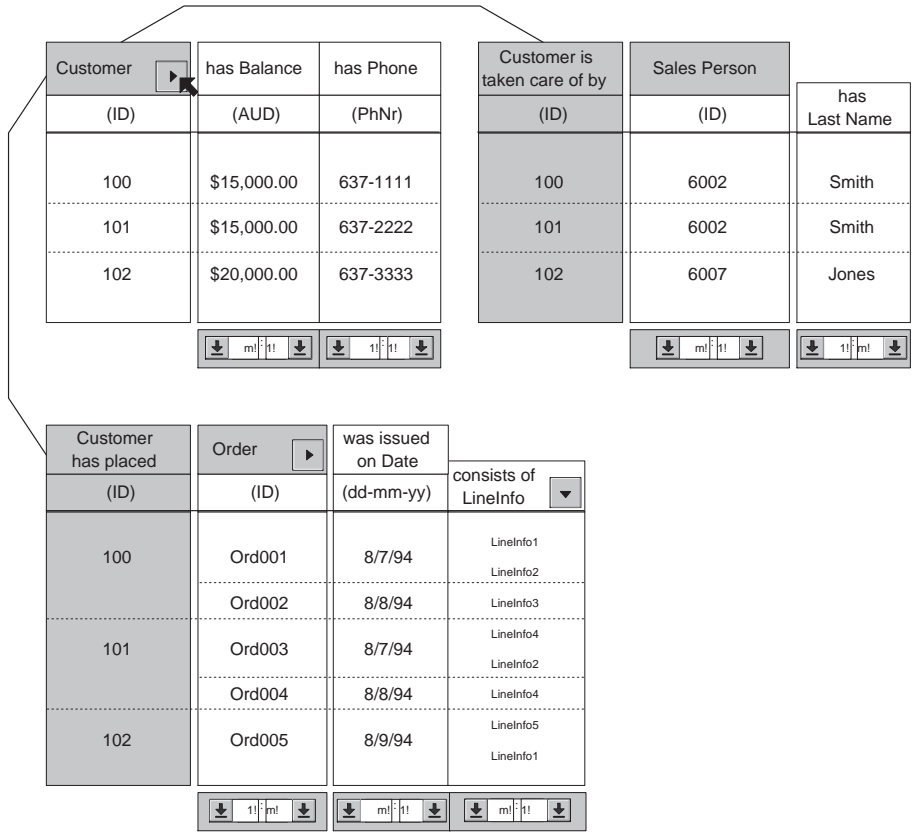
{\def\Scale{0.6} \EpsfFig[\SampleGrid]{An Example Grid}}
The example generation tool itself consists of two basic elements.
Firstly, the user must be able to select parts of a conceptual schema, and
put them on the screen in an orderly way.
These parts together form a tree\footnote{One could argue that this should
actually be a sequence of trees, however, due to the limited size of PC Screens
and user's capabilities to deal with large amount of information at the same time,
it is probably wiser to limit ourselves to a single tree of limited size.}. 
As an example, consider \SRef{\SampleGrid}.
This screen depicts four interconnected tables, and is based on the schema
depicted in \SRef{\SampleSchema}.
{\def\Scale{0.8} \EpsfFig[\SampleSchema]{Example Schema}}

In \cite{AsyReport:94:Harding:ExView}, similar examples can be found that
are discussed in more detail than we do here.
The inter-predicate (inter relationship type) uniqueness constraint shown
in \SRef{\SampleSchema} can not be handled as such by the example generation
algorithm, as the algorithm focusses on each relationship type separately 
during the generation process.
However, a natural solution appears when realising that when enforcing an
inter-predicate uniqueness constraint, this constraint is actually enforced on 
a derived relationship.
In \SRef{\IPConstraint} the inter-predicate uniqueness constraint from 
\SRef{\SampleSchema} is converted to an intra-predicate uniqueness constraint
on a derived relationship.
In \cite{PhdThesis:92:Hofstede:DataMod} and \cite{Report:90:Weide:Uniquest}
an algorithm is provided to actually construct this derived relationhip type
as part of the semantics of inter-predicate uniqueness constraints.
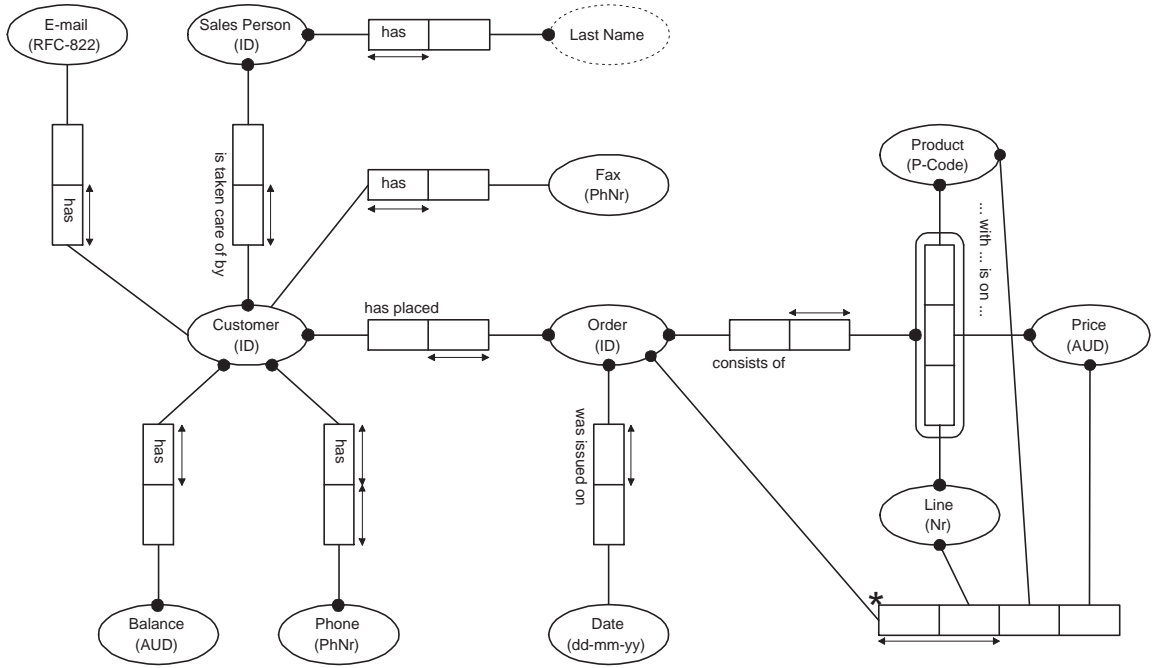
{\def\Scale{0.8} \EpsfFig[\IPConstraint]{Dealing with Inter-Predicate Constraints}}

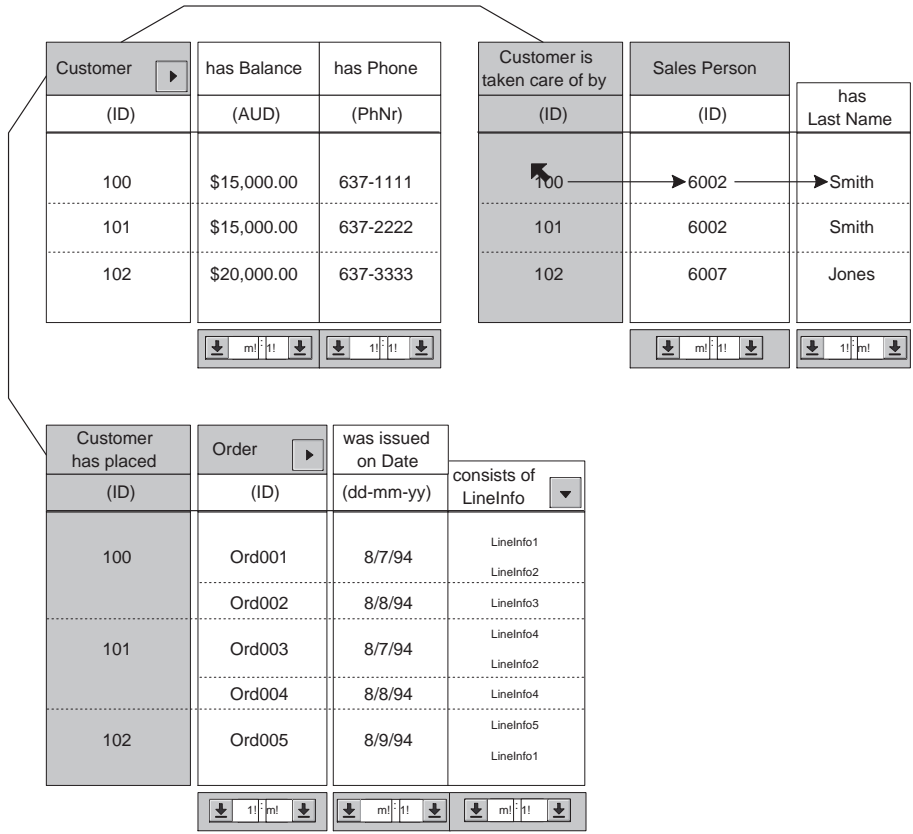
{\def\Scale{0.6} \EpsfFig[\SampleGridMouseA]{An Example Mouse Over}}
An optional feature helping the user in better understanding the structure of 
the displayed examples is illustrated in figure \ref{\SampleGridMouseA} and
\ref{\SampleGridMouseB}.
These figures illustrate a possible {\em mouse over} effect.
When the mouse cursor is over one of the values in a table, arrows are
shown illustrating the connections between the instances in the table.
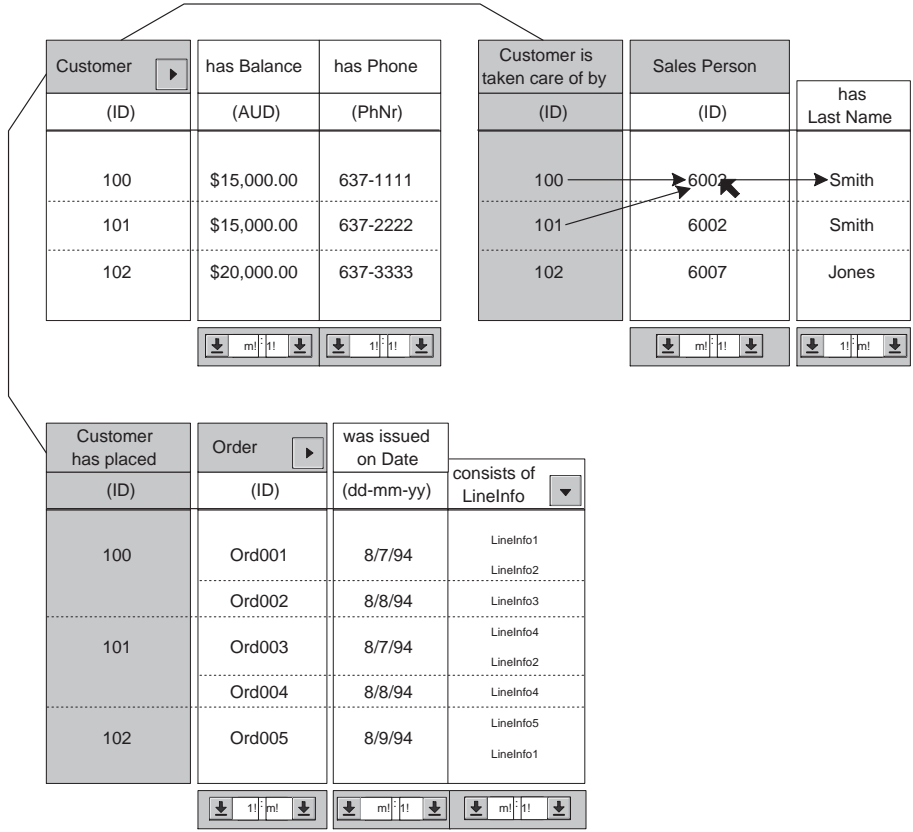
{\def\Scale{0.6} \EpsfFig[\SampleGridMouseB]{Another Mouse Over}}

{\def\Scale{0.6} \EpsfFig[\SampleGridMore]{Extendable Column}}
Two further aspects of the example in \SRef{\SampleGrid} that are noteworthy, 
as we elaborate them further in the remainder of the report, are the 
right-button following \SF{Customer} and the down-button following 
\SF{Line Info}.
The first button is used to indicate that more facts about customers are 
stored than currently shown on the screen, i.e.\/ the tree could be breathened.
This is even better illustrated in \SRef{\SampleGridMore}.
The second button indicates that \SF{Line Info} is a compositely identified 
object type (nested relationship types are regarded as compositly identified
object types as well).
Clicking on this button leads to the screen depicted in 
\SRef{\SampleGridExplode}.
In this screen the \SF{Line Info} column is split according to the reference
schema of object type \SF{Line Info}.
Note that the down-button is now replaced by an up-button to indicate that 
the details can be hidden again if so desired.
{\def\Scale{0.6} \EpsfFig[\SampleGridExplode]{Exploded View on Line Items}}

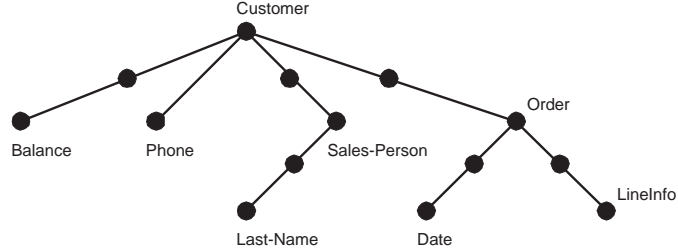
{\def\Scale{1.2} \EpsfFig[\SampleTree]{An Example Tree}}
As stated before, an example grid can be seen as a tree.
In \SRef{\SampleTree} we have depicted the tree that can be associated to 
this sample grid.
The tree underlying a sample grid is constructed by repeatedly selecting 
items from the conceptual schema and adding them to the existing tree.
Initially, there is no order provided in which the nodes of the tree 
should be put on the screen.
However, the items in the header can be shifted around by users at will. 
Alternatively, one could let the system re-shuffle the entire tree using the 
conceptual relevance (\cite{Article:94:Campbell:Abstraction}) of the object 
types involved as an ordering criterion. 
Furthermore, the spider query mechanism (\cite{AsyReport:94:Proper:SQ}) 
could be used to quickly put complete trees on the screen. 

The second element of the example generation tool is the generator of the
examples itself.
Given a tree as depicted in \SRef{\SampleTree}, a significant population for
the relationship types contained (in effect the edges) in the tree must be 
constructed.
As stated before, the significance in this first setup is limited to the 
generation of all valid combinations for the relationship types in the tree.
We limit ourselves to the intra predicate uniqueness and totality constraints
only.

The structure of this report is as follows, in \SRef{Forrest} we define 
the syntax of the trees underlying the sample grids, and provide a brief 
formalisation of the required concepts of Object-Role Modelling.
In \SRef{Negotiate} we discuss a negotiation mechanism for the size of 
the population of the edges in the trees.
The generation of the actual sample population is covered in \SRef{ExGen}.
Finally, \SRef{Concl} concludes the report.
For the reader who is unfamiliar with the notation style used in this report,
it is advisable to first read \cite{AsyReport:94:Proper:Formal}.

   \section{Creating a Forrest}
\SLabel{section}{Forrest}

This section discusses the syntax of a forrest for the example grid and also
provides a brief formalisation of the ORM concepts needed.
Although the example generation tool will initially be used in the context of
DBCreate (value types, entity types and binary relationship types), we already 
allow for ORM schemas as used by InfoModeler.
We start out from a formalisation of ORM based on the one used in 
(\cite{Report:94:Halpin:ORMPoly}). 
However, since only a limited part of the formalisation is needed, we do not 
cover the formalisation in full detail.

\subsection{ORM Basis}
\SLabel{subsection}{ORMBase}
A conceptual schema is presumed to consist of a set of types $\Types$. 
Within this set of types two subsets can be distinguished: 
   the relationship types $\RelTypes$, 
   and the object types $\ObjTypes$. 
Furthermore, let $\Preds$ be the set of roles in the conceptual schema. 
The fabric of the conceptual schema is then captured by two functions and two 
predicates. 
The set of roles associated to a relationship type are provided by the 
partition: $\Roles: \RelTypes \Func \Powerset(\Preds)$. 
Using this partition, we can define the function $\Rel$ which returns for each 
role the relationship type in which it is involved:
   $\Rel(r) = f \iff r \in \Roles(f)$.
Every role has an object type at its base called the player of the role.
This player is formally provided by the function: 
   $\Player: \Preds \Func \Types$. 
Subtyping of object types is captured by the predicates
   $\Spec \subseteq \ObjTypes \Carth \ObjTypes$. 
Using $\Spec$ we can define the notion of type relatedness: $x \TypeRel y$ for
object types $x$ and $y$.
This notion captures the intuition that two object types may share instances.
This relation is defined by the following four derivation rules:
\begin{enumerate}
   \item $x \in \Types ~\vdash~ x \TypeRel x$
   \item $x \Spec y ~\vdash~ x \TypeRel y$  
   \item $x \TypeRel y ~\vdash~ y \TypeRel x$
   \item $x \TypeRel y \TypeRel z ~\vdash~ x \TypeRel z$
\end{enumerate}
Note that when using ORM with the advanced concepts 
(\cite{Report:94:Halpin:ORMPoly}) such as polymorphism, sequence types, 
set types, etc., the definition of $\TypeRel$ needs to be refined. 

Instances of all non-value types must be identified in terms of instances of
other object types.
This identification is usually provided by a so called {\em reference schema}.
If $\ValueTypes$ denotes the set of value types, then the (direct!) 
identification relationship between types is presumed to be captured by the 
function:
\[ \RefSch: (\Types \SetMinus \ValueTypes) \Func 
            \Types \union \Preds^+ \union (\Preds \Carth \Preds)^+ \]
Each non-value type is either identified by a type (a super type), or
a sequence of roles (a relationship type), or a sequenc of role pairs 
(compositely identified object types).
Note that in this report we do not concern ourselves with well-formedness 
rules on reference schemas.
A function that is derived from $\RefSch$ and which is needed in the 
remainder is $\IdfObjs: \ObjTypes \Func \ObjTypes^+$.
This function returns the sequence of object types needed to directly identify
a given object type.
Its definition is provided as:
\[ \IdfObjs(x) ~\Eq~ 
      \left\{ \begin{array}{ll}
         [y]              & \mbox{if~} \RefSch(x) = [y]\Eol
         [y_1,\ldots,y_l] & \mbox{if~} \RefSch(x) = [\tuple{p_1,y_1},\ldots,\tuple{p_l,y_l}]\Eol
         [y_1,\ldots,y_l] & \mbox{if~} \RefSch(x) = [\tuple{p_1,q_1,y_1},\ldots,\tuple{p_l,q_l,y_l}]
      \end{array} \right. 
\]
Note that we presume the existance of an implicit coercion function 
between sequences and sets.
So, for example, if $S$ is a set of sequences we allow ourselves to 
write $\union S$ for the set of all elements occurring in the sequences 
in $S$.

A lot of the decissions made by the example generation algorithm are based
on the maximum size of the populations of the object types. 
Later on, an algorithm to calculate these sizes is presented. 
However, the maximum sizes of the populations of the value types should be
given by the modeller. 
For instance, a value type representing the gender of persons will usually
contain at most two instances.
To accommodate this, we presume the existance of the function 
$\DomSize: \ValueTypes \Func \N$.
One obvious requirement for this function is:
 $x,y \in \ValueTypes \land x \Spec y \implies \DomSize(x) \leq \DomSize(y)$.

Finally, in this report an important role is played by the uniqueness and
tolality constraints.
For that purpose, we presume the predicates $\Unique \subseteq \Powerset(\Preds)$
and $\Total \subseteq \Powerset(\Preds)$ to provide all uniqueness and totality
constraints.

In this report we thus only use the following components of an ORM schema:
\[ \tuple{\Types, \RelTypes, \ValueTypes, \ObjTypes, \Preds,
          \Spec, \Roles, \Player, \RefSch, \DomSize, \Unique, \Total} \]
When implementing the algorithms and ideas presented in this article, these 
components are the {\em interface} between the example generator tool and
the meta model from the fact base.

\subsection{Forrests}
A tree of an example grid is built from a set of nodes.
Let in the remainder $\Nodes$ be a set of all such nodes.
Formally, a tree can now be defined by two functions: $\EOut$ and $\Obj$.
As one object type can be represented by more than one node, the object type 
represented by a node is provided by the function:
   $\Obj: \Nodes \PartFunc \ObjTypes$.
The edges of the tree are provided by the function
   $\EOut: \Nodes \Func \Powerset(\Links \Carth \Nodes)$.
This function provides for each node the set of outgoing edges.
Using the $\EOut$ function the following ``inverse'' function can be derived, 
which returns for each node the set of incomming edges:
\begin{quote}
   $\EIn: \Nodes \Func \Powerset(\Links \Carth \Nodes)$\Eol
   $\EIn(n) ~\Eq~ \Set{\tuple{l,m}}{\tuple{l,n} \in \EOut(m)}$
\end{quote}
The edges of the tree are labelled with link information.
The set of links $\Links$ is defined by:
$\Links ~\Eq~ \Preds \union \RPreds$, where
$\RPreds ~\Eq~ \Set{\Rev{p}}{p \in \Preds}$ is the set of {\em reversed}
roles.
The reversed roles can be used to connect a node representing a relationship 
type to a node representing one of the participating object types.
Each link has a starting point and an ending point.
To access these points (object types) uniformly, we introduce the
following two generic functions:
\begin{quote}
   $\Start,\End: \Links \Func \Types$\EOl
   $\Start(x) ~\Eq~ \left\{\begin{array}{ll}
                      \Player(x) & \mbox{if~} x \in \Preds\Eol
                      \Rel(x)    & \mbox{otherwise}
                    \end{array} \right.$\EOl
   $\End(x) ~\Eq~ \left\{\begin{array}{ll}
                     \Rel(x)    & \mbox{if~} x \in \Preds\Eol
                     \Player(x) & \mbox{otherwise}
                  \end{array} \right.$
\end{quote}
Note: from now on we presume $\Rel$ and $\Player$ to be generalised to
elements from $\RPreds$ in the obvious way.

In order for $\EOut$ and $\Obj$ to span a tree, they must adhere to certain
properties.
Each edge in the tree must be a connection between the source and destination 
of the edge via the role. 
This is expressed by the following axiom:
\begin{Axiom}{T}[EdgeWF]
   For each $l \in \Links$: 
   \[ \tuple{l,m} \in \EOut(n) \implies 
      \Start(l) \TypeRel \Obj(n) \land \Obj(m) = \End(l) 
   \]
\end{Axiom}
The function $\EOut$ must indeed define a tree, so the graph spanned by
$\EOut$ has to be a connected, acyclic graph with a unique root.
\begin{Axiom}{T}
   $\card{\Proj_2(\union \FuncRan(\EOut))} = \card{\FuncRan(\EOut)}$
\end{Axiom}
\begin{Axiom}{T}
   $\Eu{x \in \FuncDom(\EOut)}{
       x \not\in \Proj_2(\union\FuncRan(\EOut))
    }$
\end{Axiom}
Note: if a directed graph has a unique root then it is automatically 
connected.
All used nodes must have an object type associated:
\begin{Axiom}{T}
   $\FuncDom(\EOut) \union \Proj_2(\union\FuncRan(\EOut)) \subseteq \FuncDom(\Obj)$
\end{Axiom}
Furthermore, a link can be used only once for an outgoing edge of a 
node.
This is formally enforced by:
\begin{Axiom}{T}
   \( 
      \tuple{l,n_1} \in \EOut(m) \land
      \tuple{l,n_2} \in \EOut(m) \implies n_1 = n_2
   \)
\end{Axiom}
One of the items that can be shown when displaying nodes on screen is the
fact that a given node can be extended with more edges.
In \SRef{\SampleGrid}, the right-button behind the \SF{Customer} indicated 
that more relationship types are available for this object type.
A node can be extended with an extra edge iff a link exists that can form a 
proper edge, and is not already used by on another edge of this node.
This property can be expressed formally as:
\[ \SF{CanExtend}(n) \iff 
   \Set{l \in \Preds}{
      \Start(l) \TypeRel \Obj(n) \land 
      l \not\in \Proj_1 \EOut(n) 
   } \not= \emptyset
\]
The set in the righthand side of the above definition can actually be used to 
fill the listbox displayed in \SRef{\SampleGridMore}.

The order in which the nodes themselves are displayed on the screen is 
recorded by the function $\Order: \Nodes \Func \N$. 
As this function must provide a total order of the nodes, we should have:
\[ \Order(x_1) = \Order(x_2) \implies x_1 = x_2 \]
One additional option of a system using an example generator is to have the 
system re-order the tree based on the conceptual relevance of the object 
types represented by the nodes. 
If $\CWeight: \Types \Func \N$ is a function returning the conceptual weight 
of types, then the system is able to order the nodes such that:
\[ \Order(n_1) < \Order(n_2) \implies 
   \CWeight(\Obj(n_1)) \leq \CWeight(\Obj(n_2)) \]
Note: more than one order may exist for the same $\CWeight$ values since 
differing object types may have the same conceptual weight.
Not all nodes need to be displayed on the screen.
Some nodes can be left implicit.
For instance, a node representing a non-objectified binary relationship
does not have to be shown on the screen.
The set of implicit nodes in a tree is identified by:
\[ 
   \ImplNodes ~\Eq~ 
   \Set{n \in \Nodes}{\Proj_1\EIn(n) \intersect  \RPreds = \emptyset \land
                      \Proj_1\EOut(n) \intersect \Preds  = \emptyset \land
                      \EIn(n) \neq \emptyset \land \EOut(n) \neq \emptyset} 
\]
From this definition immediately follows that two neighbouring nodes cannot
both be implicit.
This leads to the following lemma:
\begin{lemma}(no implicit neighbours)[NoImplNeighb]
   $n \in \Nodes \land m \in \Proj_2\EOut(n) \implies 
    \setje{n,m} \not\subseteq \ImplNodes$
\end{lemma}
\begin{proof}
   Let $n \in \Nodes \land m \in \Proj_2\EOut(n)$ 
   such that $\setje{n,m} \subseteq \ImplNodes$.
   
   Since $m \in \Proj_2\EOut(n)$ it immediately follows from the
   definition of $\EIn$ that:
   \[ \Proj_1\EIn(m) = \Proj_1\EOut(n) \]
   As we presumed that $n,m \in \ImplNodes$ we in particular have:
   \[
     \Proj_1\EOut(n) \intersect \Preds  = \emptyset \mbox{~and~} 
     \Proj_1\EIn(m)  \intersect \RPreds = \emptyset
   \]
   Since $\Proj_1\EIn(m) = \Proj_1\EOut(n)$ we have:
   \[ 
     \Proj_1\EOut(n) \intersect \Preds  = \emptyset \mbox{~and~} 
     \Proj_1\EOut(n) \intersect \RPreds = \emptyset
   \]
   Which implies that $\EOut(n)$ which is a contradiction 
   since $m \in \Proj_2\EOut(n)$.

   Therefore we can not have $\setje{n,m} \subseteq \ImplNodes$ if
   $n \in \Nodes \land m \in \Proj_2\EOut(n)$.
\end{proof}
As an example, consider the tree depicted on the right side of 
\SRef{\ExampleImplicits} in the context of the schema shown there as well.
The open circles represent nodes that do not have to be shown on the screen
when this tree is shown to the user.
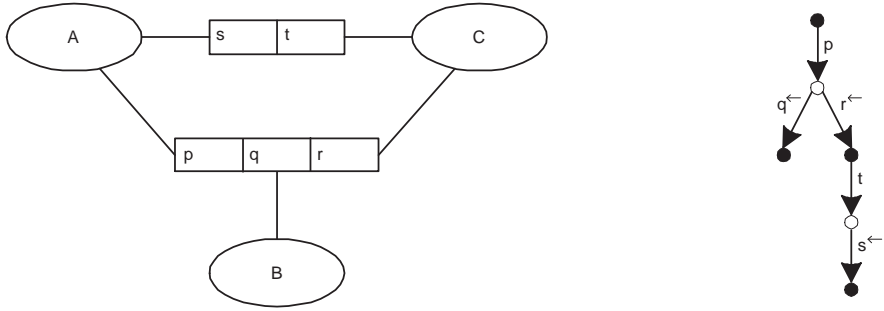
\EpsfFig[\ExampleImplicits]{Examples of Implicit Nodes}

A further result is the following lemma stating that each implicit node
must correspond to a relationship type, and that the predicators used to label the
incoming and outgoing edges are all of the same relationship type:
\begin{lemma}(single relationship type involvement)[ImplNodesRelTypes]
   $n \in \ImplNodes \implies 
      \Al{l \in \Proj_1 (\EOut(n) \union \EIn(n))}{\Rel(l) = \Obj(n)}$
\end{lemma}
\begin{proof}
   If $n \in \ImplNodes$ then we must, due to the definition of $\ImplNodes$ 
   have $l \in \Proj_1 \EIn(n) \implies l \in \Preds$.
   This implies that $\Rel(l) = \Obj(n)$.
   So $\Al{l \in \Proj_1 \EIn(n)}{\Rel(l) = \Obj(n)}$

   Furthermore, if $l \in \Proj_1 \EOut(n)$, it follows from the definition
   of $\ImplNodes$ that $l \in \RPreds$.
   From \SRef{EdgeWF} follows that $\Start(l) \TypeRel \Obj(n)$, which 
   can be reformulated as $\Rel(l) \TypeRel \Obj(n)$.
   Since two relationship types can, in ORM, only be type related if they
   are the same, we therefore have: $\Rel(l) = \Obj(n)$.
   As a result: $\Al{l \in \Proj_1 \EOut(n)}{\Rel(l) = \Obj(n)}$.
\end{proof}
To cater for identification, and in particular complex identification,
nodes can have associated a number of identifying nodes.
This is captured by the function 
\[ \NRefSch: \Nodes \PartFunc 
             \Nodes \union
             (\Preds \Carth \Nodes)^+ \union 
             (\Preds \Carth \Preds \Carth \Nodes)^+ \]
which can defines an (possibly empty) identification tree for each node in 
the tree provided by $\EOut$.
This function must always behave conform the identification given in 
the schema:
\begin{Axiom}{T}[IdfConformity]
   For each $n \in \FuncDom(\NRefSch)$ we should have:
   \begin{enumerate}
      \item $\NRefSch(n) = [m] \implies \RefSch(\Obj(n)) = [\Obj(m)]$
      \item $\NRefSch(n) = [\tuple{p_1,m_1},\ldots,\tuple{p_l,m_l}] \implies$\Eol
            \Tab \Tab 
            $\RefSch(\Obj(n)) = [p_1,\ldots,p_l] \land
             \Al{1 \leq i \leq l}{\Player(p_i) = \Obj(m_i)}$
      \item $\NRefSch(n) = [\tuple{p_1,q_1,m_1},\ldots,\tuple{p_l,q_l,m_l}] \implies$\Eol
            \Tab \Tab
            $\RefSch(\Obj(n)) = [\tuple{p_1,q_1},\ldots,\tuple{p_l,q_l}] \land
             \Al{1 \leq i \leq l}{\Player(q_i) = \Obj(m_i)}$
   \end{enumerate}
   where $m,m_1, \ldots, m_l \in \Nodes$
   and $p_1,\ldots,p_l,q_1,\ldots,q_l \in \Preds$.
\end{Axiom}
The above axiom can thus be seen as an invariance requirement on the 
algorithm that builds the tree of the example grid.

Not all nodes used by $\EOut$ must necessarily have an identification
tree associated.
For instance, in the first example grid of the running example, 
\SF{Line Info} instances are not denoted by means of their full 
identification.
We used textual representations (surogates) of abstract instances such as: 
\SF{LineInfo1}, \SF{LineInfo2}, etc.

For the remaining axioms on identification trees in the example grids, we need
one more derived function.
The function $\IdfNodes: \Nodes \Func \Nodes^+$ which (analogously to
$\IdfObjs$) determines the set of nodes needed to directly identify a given 
node. 
Its definition is provided as:
\[ \IdfNodes(x) ~\Eq~ 
      \left\{ \begin{array}{ll}
         [y]              & \mbox{if~} \NRefSch(x) = [y]\Eol
         [y_1,\ldots,y_l] & \mbox{if~} \NRefSch(x) = [\tuple{p_1,y_1},\ldots,\tuple{p_l,y_l}]\Eol
         [y_1,\ldots,y_l] & \mbox{if~} \NRefSch(x) = [\tuple{p_1,q_1,y_1},\ldots,\tuple{p_l,q_l,y_l}]
      \end{array} \right. 
\]
Note again that we presume the existance of an implicit coercion function 
between sequences and sets.

We can now require each identification tree to be a tree indeed.
\begin{Axiom}{T}(acyclic)
   If $x \in \FuncDom(\NRefSch)$, then we have 
   $\SF{NonCyclic}(x,\setje{x})$, where:
   \[ \SF{NonCyclic}(x,X) ~\iff~
      X \intersect \IdfNodes(x) = \emptyset \land
      \Al{y \in \IdfNodes(x)}{\SF{NonCyclic}(y,X \union \setje{y})}
   \]
\end{Axiom}
Note that for simple ORM schemas the fact that $\NRefSch$ is a tree
for each object type in the tree spanned by $\EOut$, follows directly from
\SRef{IdfConformity} and the acyclicity of the identification trees spanned
by $\RefSch$. 
However, when using the polymorphism concept in more advanced ORM models, in
particular when defining recursive data structures, $\RefSch$ does not
necessarily have to span trees anymore (but $\NRefSch$ is still required to
do so).

For obvious reasons, all nodes used in the identification trees have an 
object type associated:
\begin{Axiom}{T}
    $\union\FuncRan(\IdfNodes) \subseteq \FuncDom(\Obj)$
\end{Axiom}
The example grid tree and the identification trees should not 
be intermixed (except for the roots of the identification trees):
\begin{Axiom}{T}
   $\union\FuncRan(\IdfNodes) \intersect 
    (\FuncDom(\EOut) \union \FuncDom(\EIn)) = \emptyset$
\end{Axiom}
The root nodes of the identification trees should indeed be part
of the tree for the example grid:
\begin{Axiom}{T}
   $(\FuncDom(\IdfNodes) \SetMinus \union\FuncRan(\IdfNodes)) \subseteq 
     (\FuncDom(\EOut) \union \FuncDom(\EIn))$
\end{Axiom}
Finally, one good default rule is to automatically add simple identifications.
So the following rule should be an invariant when manipulating the trees
(e.g. when adding a simply identified object type):
\begin{Axiom}{T}
   $\card{\RefSch(\Obj(x))} = 1 \implies \DefinedAt{\NRefSch}{x}$
\end{Axiom}
A tree for the example grid, in the context of an ORM schema, is now 
completely determined by the following five components:
\[ \tuple{\Nodes,\EOut,\Obj,\Order,\NRefSch} \]

   \section{Negotiating the Size of the Example Space}
\SLabel{section}{Negotiate}

As the title of this section suggests, in this section we concern ourselves
with a negotiation mechanism to determine the number of examples that are
to be used in the example grid.
Such a negotation phase is needed because some object types may have a limited 
set of instances. 
In particular value types such as \SF{Gender} which will generally have two 
or three instances only. 

The first question we need to answer is the maximum number of instances that 
a given type may have.
Due to relationships between types and the constraint patterns associated to
the relationships; limiting the number of instances of a value type may propagate
through the conceptual schema.
As an example consider the schema shown in \SRef{\SizePropag}.
Object type \SF{B} has a maximum number of instances of 2. 
As each instance of object type \SF{A} must play relationship \SF{f} with 
a unique instance of \SF{B}, there can be only two instances of \SF{A}.
As a result, there can be only two instances of relationship \SF{f}, so
if \SF{f} would be objectified this could lead to another propagation of
a size limitation.
In the context of larger schemas, these propagations may cause ``ripple'' 
effects through the entire schema. 
Note that if \SF{A} is connected to \SF{B} through a series of other 
relationship types, then it could even be the case that the maximum size 
of \SF{B} needs to be reduced further! 
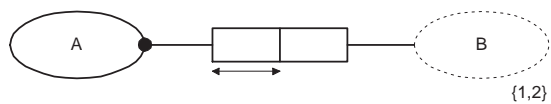
\EpsfFig[\SizePropag]{Propagating Maximum Size of Populations}

Initially, object types are presumed to have a potentially infinite maximum
population.
For this purpose we need to introduce the notion of {\em infinity} as an 
explicit element in our calculations.
Let $\infty$ denote infinity, then our population size calculations  
take place within the set: $\NInf \Eq \N \union \setje{\infty}$.
For $\NInf$ we inherit the $+$, $\times$ and $<$ operations from $\N$ with the 
following additional cases:
\begin{quote}
  if $n \in \N$, then:\Eol
  \Tab $n \leq \infty$\Eol
  \Tab $n + \infty = \infty + n = \infty$\Eol
  \Tab $n \times \infty = \infty \times n = \infty$
\end{quote}
As the minimum $\Min(n,m)$ of two natural numbers is defined in terms of 
$<$, we obviously have:
$n \in \N \implies \Min(n,\infty) = \Min(\infty,n) = n$.

\subsection{Generating Patterns}
In determining the maximum sizes of populations, we need to generate for
each relationship type the significant combinations of instances.
As stated before, we consider a combination of instances to be significant
if it shows all allowed combinations with respect to the cardinality constraints
defined on that relationship type.
The patterns are generated by the algorithm given below\footnote{I would 
like to thank L. Campbell for providing me with the first informal draft 
version of this algorithm.}.
This algorithm takes as input the relationship type to be populated and
the maximum sizes of the types (determined so far); in particular the
players of the roles in the relationship type and the maximum size of
the relationship type itself. 
The latter size is relevant in the case of objectified relationship types.
Suppose in the example of \SRef{\SizePropag}, \SF{A} is actually an objectified
relationship type.
In such a case we can only generate two instances for relationship type \SF{A}.
The algorithm itself is given as:
\begin{quote}
   $\GenPattern: \RelTypes \Carth 
                 (\Types \Func \NInf)
                 ~\Func~ 
                 (\Types \PartFunc \N) \Carth
                 \Powerset(\Preds \Func \N)$\Eol
   $\GenPattern(\Vr{Rel},\Vr{Size}) ~\Eq~$\Eol
   \Tab \VAR\Eol
   \Tab \Tab \begin{tabular}{ll}
                \Vr{Pattern}:    & $\Powerset(\Preds \Func \N)$;\Eol
                \Vr{FreshTuple}: & $\Preds \Func \N$;\Eol
                \Vr{WorkTuple}:  & $\Preds \Func \N$;\Eol
                \Vr{Used}:       & $\Types \PartFunc \NInf$;\Eol
                $p$:             & $\Preds$;
             \end{tabular}\Eol
   \Eol
   \Tab \MACRO\Eol
   \Tab \Tab $\Vr{Extendable}(P: \Powerset(\Preds)) \equiv$\Eol 
   \Tab \Tab \Tab $\Al{p \in P}{\Vr{Used}(\Player(p)) < \Vr{Size}(\Player(p))} ~\land$\Eol
   \Tab \Tab \Tab $\Vr{Used}(\Rel(p)) + \card{P} \leq \Vr{Size}(\Rel(p))$; \Eol
   \Eol  
   \Tab \Tab $\Vr{IncrUsed}(p: \Preds) \equiv$\Eol
   \Tab \Tab \Tab \BEGIN\Eol
   \Tab \Tab \Tab \Tab $\Vr{Used}(\Player(p)) \Pab 1$;\Eol
   \Tab \Tab \Tab \Tab $\Vr{Used}(\Rel(p)) \Pab 1$;\Eol
   \Tab \Tab \Tab \END;\Eol
   \Eol
   \Tab \BEGIN\Eol
   \Tab \Tab \SF{\# Initialise variables}\Eol
   \Tab \Tab \Vr{Pattern} := $\emptyset$;\Eol
   \Tab \Tab \Vr{Used}(\Vr{Rel}) := 0;\Eol
   \Tab \Tab \FOREACH $p \in \Roles(\Vr{Rel})$ \DO\Eol
   \Tab \Tab \Tab \Vr{Used}($\Player(p)$) := 0;\Eol
   \Tab \Tab \ENDFOR;\Eol
   \Eol
   \Tab \Tab \WHILE $\Vr{Extendable}(\Roles(\Vr{p}))$ \DO\Eol
   \Tab \Tab \Tab \SF{\# Generate fresh tuple}\Eol
   \Tab \Tab \Tab \FOREACH $p \in \Roles(\Vr{Rel})$ \DO\Eol
   \Tab \Tab \Tab \Tab $\Vr{IncrUsed}(p)$;\Eol
   \Tab \Tab \Tab \Tab $\Vr{FreshTuple}(p) := \Vr{Used}(\Player(p))$;\Eol
   \Tab \Tab \Tab \ENDFOR;\Eol
   \Eol
   \Tab \Tab \Tab \SF{\# Probe uniqueness}\Eol
   \Tab \Tab \Tab $\Vr{Pattern} \Pab \setje{\Vr{FreshTuple}}$;\Eol
   \Tab \Tab \Tab \FOREACH $p \in \Roles(\Vr{Rel})$ \DO\Eol
   \Tab \Tab \Tab \Tab \Vr{WorkTuple} := \Vr{FreshTuple};\Eol 
   \Eol
   \Tab \Tab \Tab \Tab \SF{\# Try to mutate tuple}\Eol
   \Tab \Tab \Tab \Tab \IF $\lnot\Ex{\tau \subseteq \Roles(\Vr{Rel})}{
                                \Unique(\tau) \land p \not\in \tau
                            } \land
                            \Vr{Extendable}(\setje{p})$ 
                       \THEN\Eol
   \Tab \Tab \Tab \Tab \Tab $\Vr{IncrUsed}(p)$;\Eol
   \Tab \Tab \Tab \Tab \Tab $\Vr{WorkTuple}(p) := \Vr{Used}(\Player(p))$;\Eol
   \Tab \Tab \Tab \Tab \Tab $\Vr{Pattern} \Pab \setje{\Vr{WorkTuple}}$;\Eol
   \Tab \Tab \Tab \Tab \FI;\Eol
   \Tab \Tab \Tab \ENDFOR;\Eol
   \Eol
   \Tab \Tab \Tab \SF{\# Generate nil tuple}\Eol
   \Tab \Tab \Tab \FOREACH $p \in \Roles(\Vr{Rel})$ \DO\Eol
   \Tab \Tab \Tab \Tab $\Vr{IncrUsed}(p)$;\Eol
   \Tab \Tab \Tab \Tab $\Vr{FreshTuple}(p) := 0$;\Eol
   \Tab \Tab \Tab \ENDFOR;\Eol
   \Eol
   \Tab \Tab \Tab \SF{\# Probe totality}\Eol
   \Tab \Tab \Tab \FOREACH $p \in \Roles(\Vr{Rel})$ \DO\Eol
   \Tab \Tab \Tab \Tab \Vr{WorkTuple} := \Vr{FreshTuple};\Eol
   \Eol 
   \Tab \Tab \Tab \Tab \SF{\# Try to mutate tuple}\Eol
   \Tab \Tab \Tab \Tab \IF $\lnot\Total(p) \land \Vr{Extendable}(\setje{p})$ \THEN\Eol
   \Tab \Tab \Tab \Tab \Tab $\Vr{IncrUsed}(p)$;\Eol
   \Tab \Tab \Tab \Tab \Tab $\Vr{WorkTuple}(p) := \Vr{Used}(\Player(p))$;\Eol
   \Tab \Tab \Tab \Tab \Tab $\Vr{Pattern} \Pab \setje{\Vr{WorkTuple}}$;\Eol
   \Tab \Tab \Tab \Tab \FI;\Eol
   \Tab \Tab \Tab \ENDFOR;\Eol
   \Tab \Tab \WEND;\Eol
   \Eol
   \Tab \Tab \RETURN $\tuple{\Vr{Used},\Vr{Pattern}}$;\Eol
   \Tab \END.
\end{quote} 

\subsection{Schema Plausibility Check}
An interesting spin-off of the pattern generation algorithm is that when
using this algorithm to determine the maximum size of all types in a
conceptual schema, the result can be used to do a plausibility on the 
conceptual schema.
If some type has a number maximum number of instances of $0$, it is highly 
likely that there is an error in the constraint patterns of the conceptual 
schema.
Since the pattern generation algorithm only uses a limited class of patterns,
this kind of plausibility check may be a bit ``oversensitive''.
 
When determining the maximum sizes of all types in a conceptual schema, the
results of the pattern generation algorithm can be used to refine, the current 
maximum sizes of the object types can be resized.
This resizing is done by the algorithm below.
It (re-)calculates the maximum sizes of the object types involved in each 
relationship type for which one of the involved types already has a maximum
size that is not infinite.
\begin{quote}
   $\ReSize: (\Types \Func \NInf) \Func 
             (\Types \Func \NInf)$\Eol
   $\ReSize(\Vr{Size}) ~\Eq~$\Eol
   \Tab \VAR\Eol
   \Tab \Tab \begin{tabular}{ll}
                \Vr{Pattern}: & $\Powerset(\Preds \Func \N)$;\Eol
                \Vr{Used}:    & $\Types \PartFunc \NInf$;\Eol
                \Vr{ToDo}:    & $\Powerset(\RelTypes)$;\Eol
                $x,r$:        & $\Types$;\Eol
                $p$:          & $\Preds$;
             \end{tabular}\Eol
   \Eol
   \Tab \BEGIN\Eol
   \Tab \Tab \SF{\# We should only take relationship types with at least one finite}\Eol
   \Tab \Tab \SF{\# size for the types involved into consideration. Otherwise, $\GenPattern$}\Eol
   \Tab \Tab \SF{\# would not terminate.}\Eol
   \Tab \Tab $\Vr{ToDo} := \Set{\Rel(p)}{\Vr{Size}(\Player(p)) \neq \infty \lor \Vr{Size}(\Rel(r)) \neq \infty}$;\Eol
   \Tab \Tab \SF{\# One way to optimise this further is to restrict the \Vr{ToDo} set}\Eol
   \Tab \Tab \SF{\# to those relationships for which one of the involved sizes has}\Eol
   \Tab \Tab \SF{\# changed during the last (or initial!) iteration.}\Eol
   \Eol
   \Tab \Tab \SF{\# Recalculate maximum sizes}\Eol
   \Tab \Tab \FOREACH $r \in \Vr{ToDo}$ \DO\Eol
   \Tab \Tab \Tab $\tuple{\Vr{Used},\Vr{Pattern}} := \GenPattern(r,\Vr{Size})$\Eol
   \Tab \Tab \Tab \FOREACH $x \in \Set{\Player(p)}{p \in \Roles(r)} \union \setje{r}$ \DO\Eol 
   \Tab \Tab \Tab \Tab $\Vr{Size}(x) := \Min(\Vr{Size}(x),\Vr{Used}(x))$;\Eol
   \Tab \Tab \Tab \ENDFOR;\Eol
   \Tab \Tab \ENDFOR;\Eol
   \Eol
   \Tab \Tab \SF{\# Propagate maximums in subtype hierarchy}\Eol
   \Tab \Tab \FOREACH $x,y \in \ObjTypes$ \DO\Eol
   \Tab \Tab \Tab \IF $x \Spec y$ \THEN\Eol
   \Tab \Tab \Tab \Tab $\Vr{Size}(x) := \Min(\Vr{Size}(y),\Vr{Size}(x))$;\Eol
   \Tab \Tab \Tab \FI;\Eol
   \Tab \Tab \ENDFOR;\Eol
   \Eol
   \Tab \Tab \RETURN \Vr{Size};\Eol
   \Tab \END.
\end{quote}
The idea is now to constantly call the above resize algorithm until the 
eventual maximum sizes have been stabalised.
We do this by means of the following fixed-point calculation:
\begin{quote}
   $\CalcSizes: \Func (\Types \Func \NInf)$\Eol
   $\CalcSizes() ~\Eq~$\Eol
   \Tab \VAR\Eol
   \Tab \Tab \begin{tabular}{ll}
                \Vr{Size}:    & $\Types \Func \NInf$;\Eol
                \Vr{NewSize}: & $\Types \Func \NInf$;\Eol
             \end{tabular}\Eol
   \Eol
   \Tab \BEGIN\Eol
   \Tab \Tab $\Vr{NewSize} := \MaxSize$;\Eol
   \Eol
   \Tab \Tab \SF{\# Do fixed point calculation}\Eol
   \Tab \Tab \REPEAT\Eol
   \Tab \Tab \Tab $\Vr{Size} := \Vr{NewSize}$;\Eol
   \Tab \Tab \Tab $\Vr{NewSize} := \ReSize(\Vr{Size})$;\Eol
   \Tab \Tab \UNTIL $\Vr{NewSize} = \Vr{Size}$;\Eol
   \Tab \END.
\end{quote}
Note that this algorithm terminates since the number of type classes is
finite and when resizing the maximum sizes we always choose the minimum 
value. 
In this last algorithm, $\MaxSize$ is the initial setting of the maximum
sizes.
This initial setting completely depends on the number of instances 
(or generatable) for the value types in the conceptual schema.
In \SRef{ORMBase} we introduced the function $\DomSize: \ValueTypes \Func \N$ 
as the function that provides these sizes.
From this we can derive the initial maximum size for any object type as
follows:
\begin{quote}
   $\MaxSize: \Types \Func \NInf$\EOl
   $\MaxSize(x) ~\Eq~
      \left\{\begin{array}{ll}
         \DomSize(x) & \mbox{if~} \Ex{x}{\DefinedAt{\DomSize}{x}}\Eol
         \infty      & \mbox{otherwise}
      \end{array} \right.$
\end{quote}
We are now in a position to interpret the results from $\CalcSizes$.
Let $\Size(x)$ be the size of type $x$ resulting after $\CalcSizes$.
If $\Size(x) = 0$ for some type $x$, then it is highly likely that there is 
a problem with the constraint patterns. 
In such a case, each relationship type with a player $y$ such that 
$\Size(y) = 0$ and a player $z$ such that $\Size(z) > 0$ should be examined.

Furthermore, if there is a type $x$ such that $\Size(x) < \MaxSize(x)$, there
is a limited likelyhood that there is a problem with the constraint patterns.
In this case, each relationship with a player $y$ such that 
$\Size(y) < \MaxSize(y)$ and a player $z$ such that $\Size(z) = \MaxSize(z)$,
should be examined.

\subsection{The Example Space}
When generating examples for an example grid, we do this for each 
``umbrella'' in the tree, i.e.\/ a node and its direct descendants.
The umbrella associated to a node $n$ is defined formally by:
\begin{quote}
   $\Umbr: \Nodes \Func \Powerset(\Nodes)$\Eol
   $\Umbr(n) ~\Eq~ 
      \Proj_2 \EOut(n) \union 
      \Union\Sub{y \in \Proj_2 \EOut(n) \intersect \ImplNodes} \Umbr(y)$
\end{quote}
From \SRef{NoImplNeighb} follows that an umbrella can be at most three
nodes deep: the root, one layer of explicit or implicit nodes, and the
layer containing the off-spring of the implicit nodes of layer 2.
An example of two umbrellas is shown in \SRef{\Umbrella}.
These two umbrellas correspond to two tables of the example grid shown
in \SRef{\SampleGrid}.
When generating the examples for an entire tree, we only have to consider
all umbrellas of the nodes which are not implicit, since the implicit nodes
are already contained in these umbrellas.
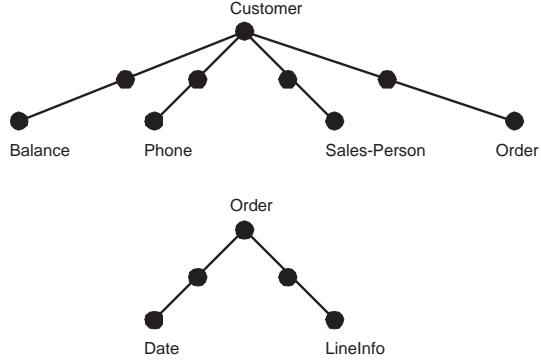
{\def\Scale{1.2} \EpsfFig[\Umbrella]{Two example umbrellas}}

The set of relationships involved in an umbrella can now simply be derived 
as follows:
\begin{quote}
   $\RelSet: \Nodes \Func \Powerset(\RelTypes)$\Eol
   $\RelSet(n) ~\Eq~ \Set{\Rel(l)}{l \in \Proj_1 \EOut(n)}$
\end{quote}
where $\Rel$ is presumed to be generalised to elements of $\RPreds$.
Note that we only need to look at the set of outgoing edges from root of 
the umbrella as the outgoing edges from the nodes of the second layer
are all implicit nodes, which allows us to apply \SRef{ImplNodesRelTypes}.

A crucial first step in the generation of the examples is the negotiation
of the number of instances in the root of each umbrella.
In this negotiation we disregard the context of each umbrella, so we can
not simply use the results of the $\CalcSizes$ function.
The reason to ignore the context lies in the fact that when validating the
constraint pattern of one relationship type, one does not want to have a
negative influence on the number of examples by a (possible) problem in
another part of the schema.
In the definition, again an initial maximum size of type populations
is needed.
From this we can derive the maximum size for any object type using the 
following recursive function:
\begin{quote}
   $\MaxSize: \ObjTypes \Func \N$\EOl
   $\MaxSize(x) ~\Eq~
      \left\{\begin{array}{ll}
         \DomSize(x) & \mbox{if~$x \in \ValueTypes$}\Eol
         \Pi\Sub{y \in \IdfObjs(x)} \MaxSize(y) & \mbox{otherwise} 
      \end{array} \right.$
\end{quote}
Using this function the actual `negotiation' function can be defined,
determining the proper size of the root of an umbrella:
\begin{quote}
   $\NodeSize: \Nodes \Func \N$ \Eol 
   $\NodeSize(n) ~\Eq~ 
     \Min({\it MaxUserSizePref},\Min\Sub{r \in \RelSet(n)} \Vr{Usage}(r,n))$
\end{quote}
where $\Vr{Usage}$ is defined as:
\begin{quote}
   $\Vr{Usage}: \RelTypes \Carth \Nodes \Func \N$\Eol
   $\Vr{Usage}(r,n) ~\Eq~ 
    \left\{ \begin{array}{ll}
       \Vr{Used}(\Obj(n)) & \mbox{if~} \MaxSize(n) \not= \infty\Eol
       \infty             & \mbox{otherwise}
    \end{array} \right.$\Eol
   such that: 
   $\tuple{\Vr{Used},\Vr{Pattern}} = \GenPattern(r,\MaxSize)$
\end{quote}
and \Vr{MaxUserSizePref} is a user defined constant providing the preferred
maximum size of root nodes. 
Note: we presume $\Max\Sub{x \in X}f(x)$ returns $0$ if $X = \emptyset$.

   \section{Filling an Example Grid}
\SLabel{section}{ExGen}

In this section we do the final step and actually populate the example
grids.
Thus far the only result from the $\GenPattern$ algorithm we used was the
usage of instances.
In this section, the patterns generated by $\GenPattern$ are finally utilised
to generate the real examples.
 
\subsection{Filling Object Types}
When generating instances for the edges of an umbrella, we need to generate
instances for the nodes (object types) first.
Let $\GenInst: \ObjTypes \Carth \N \Func \Omega$ be a given function that 
generates instances for object types.
This function need the natural number to generate unique (and yet 
deterministic) instances.

When applied for a value type $x$, the result of $\GenInst(x,n)$ is one of the
instances of the domain for values of $x$.
Suppose that the user has given 8/7/94, 8/8/94, 8/9/94 as examples for the 
\SF{Date} value, then $\GenInst(\SF{Date},2)$ would result in 8/8/94.
When the $\GenInst(x,n)$ is used for a non-value type, a surrogate instance
is generated.
For instance, $\GenInst(\SF{LineInfo},1)$ could result in $\SF{LineInfo1}$.
The set $\UniDom$ is used in this report as a general domain for simple
instances from value types and composed instances from compositely identified
types.
Using this primitive function, the following more refined generation function
can be defined, which also takes the required details of identifications 
($\IdfNodes$) into account:
\begin{quote} 
   $\GetInst: \ObjTypes \Carth \N \Func \UniDom$\Eol
   $\GetInst(x,n) ~\Eq~
    \left\{ \begin{array}{ll}
       {\it nil}                  & \mbox{if~} n=0\Eol 
       \Compose(\IdfNodes(x),n-1) & \mbox{if~} \DefinedAt{\IdfNodes}{x}\Eol
       [\GenInst(x,n)]            & \mbox{otherwise}
    \end{array} \right.$
\end{quote}
Note that the $n=0$ case is a special case.
Generating a nil value for $n=0$ allows us to treat patterns resulting from 
optional roles uniform to the other patterns.
The $\Compose$ function is used to construct the actual instance in the
case of an object type with an identification provided by $\IdfNodes$.
\begin{quote}
   $\Compose: \ObjTypes^+ \Carth \N \Func \UniDom$\Eol 
   $\Compose([x_1,\ldots,x_n],m) ~\Eq~
    \left\{ \begin{array}{ll}
       [\GetInst(x_1,\gamma(m)+1)] & \mbox{if~} n=1\EOl
       [\Compose([x_1],\gamma(m) \SFOp{DIV} \MaxSize(x_1))] 
        \SeqConc & \Eol
       \Tab [\Compose([x_2,\ldots,x_n],\gamma(m) \SFOp{MOD} \MaxSize(x_1))] &
       \mbox{otherwise}
    \end{array} \right.$\EOL
   where:
   \mbox{~~~} $\gamma(m) = \left\{ \begin{array}{ll} 
                              x/2 & \mbox{if~} \SF{IsEven(x)} \Eol
                              \Pi\Sub{1 \leq i \leq n}\MaxSize(x_i) - 
                              (m+1)/2 & \mbox{otherwise}
                           \end{array} \right.$
\end{quote}
Note that the use of the $\gamma$ function causes a ``natural'' spread in the 
use of the examples.

\subsection{Filling an Example Grid}
For filling an umbrella with instances all that remains to be done is to call
$\GenPattern$ for each relationship type contained in the umbrella with the 
negotiated nodesizes given in $\NodeSize$, 
generate the concrete instances using $\GetInst$, and finally order the
result.
All this is done by the $\GenPop$ algorithm.
The algorithm takes the root node of the current umbrella as its input 
parameter and results in a population of the relationship types contained in 
the umbrella, together with an {\em ordered} population of the object types 
playing a role in these relationship types.
The results can than be used (together with the earlier discussed $\Order$ 
for the nodes) to present the examples in an ordered way.
The algorithm is defined as:
\begin{quote}
   $\GenPop: \Nodes \Func (\RelTypes \PartFunc \Powerset(\Preds \PartFunc \UniDom)) \Carth
                   (\ObjTypes \PartFunc \UniDom^+)$\Eol
   $\GenPop(n) ~\Eq~$\Eol
   \Tab $\VAR$\Eol
   \Tab \Tab $\begin{array}{ll}
                 \Vr{RPop}:     & \RelTypes \PartFunc \Powerset(\Preds \PartFunc \UniDom)\Eol
                 \Vr{OPop}:     & \ObjTypes \PartFunc \UniDom^+\Eol
                 \Vr{Rel}:      & \RelTypes\Eol
                 \Vr{Tuple}:    & \Preds \PartFunc \UniDom\Eol
                 \Vr{NewTuple}: & \Preds \PartFunc \N\Eol
                 \Vr{Role}:     & \Preds
              \end{array}$\Eol
   \Eol
   \Tab \BEGIN\Eol
   \Tab \Tab \FOREACH $\Vr{Rel} \in \RelSet(n)$ \DO\Eol
   \Tab \Tab \Tab $\tuple{\Vr{Used},\Vr{Pattern}} := \GenInst(\Vr{Rel},\NodeSize)$;
   \Eol
   \Tab \Tab \Tab \FOREACH $\Vr{Tuple} \in \Vr{Pattern}$ \DO\Eol
   \Tab \Tab \Tab \Tab \SF{\# We presume that $\N \subseteq \UniDom$}\Eol
   \Tab \Tab \Tab \Tab \Vr{NewTuple} := \Vr{Tuple}\Eol
   \Eol
   \Tab \Tab \Tab \Tab \FOREACH $\Vr{Role} \in \Roles(\Var{Rel})$ \DO\Eol
   \Tab \Tab \Tab \Tab \Tab $\Vr{NewTuple}(\Vr{Role}) := \GetInst(\Player(\Vr{Role}),\Vr{Tuple}(\Vr{Role}))$;\Eol
   \Tab \Tab \Tab \Tab \Tab $\Vr{OPop}(\Player(\Vr{Role})) \Seqpab [\Vr{NewTuple}(\Vr{Role}))]$;\Eol
   \Tab \Tab \Tab \Tab \ENDFOR;\Eol
   \Tab \Tab \Tab \Tab $\Vr{RPop}(\Vr{Rel}) \Pab \setje{\Vr{NewTuple}}$;\Eol
   \Tab \Tab \Tab \ENDFOR;\Eol
   \Tab \Tab \ENDFOR;\Eol
   \Eol
   \Tab \Tab $\Vr{OPop} := \ReOrder(\Vr{RPop},\Vr{OPop},\Obj(n))$;
   \Eol
   \Tab \Tab $\RETURN \tuple{\Vr{RPop},\Vr{OPop}}$;\Eol
   \Tab \END.
\end{quote}
The $\ReOrder: (\RelTypes \PartFunc \Powerset(\Preds \PartFunc \UniDom)) \Carth
               (\ObjTypes \PartFunc \UniDom^+) \Carth
               \Types ~~\Func~~
               (\ObjTypes \PartFunc \UniDom^+)$
function is used to order the instances in the result.
The order is based on the number of tuples in the relationship examples that 
use the instances.
We do not provide an ordering algorithm for this function and leave that to
the choice of the programmers.
However, the following ordering condition must be met.
If $P = \ReOrder(\Vr{RPop},\Vr{OPop},x)$, then:
\[
   \Al{i,j \in \FuncDom(P)}{~
       i < j \implies 
          \card{\Tuples(\Vr{RPop},P[i],x)} \geq 
          \card{\Tuples(\Vr{RPop},P[j],x)}
   ~}
\]
where $\Tuples(\Vr{RPop},v,x) = 
       \Set{t}{\Ex{p: \Player(p)=x}{t \in \Vr{RPop}(\Rel(p)) \land t(p) = v}}$.

   \section{Conclusions}
\SLabel{section}{Concl}

In this report we have presented an algorithm to generate examples for a
selection of relationship types, such that the examples are significant
(to a certain degree).
These examples can be used to fill example grids, allowing users to validate
cardinality constraints by looking at example patterns.

As a spin-off, we also discussed the option of calculating the maximum sizes
of all types in a conceptual schema.
This would allow us to detect possible problems with cardinality constraints.

The next step is to integrate these ideas into DBCreate and in a later stage
InfoModeler itself.
The definitions in this report are already suited for the ORM schemas as they
can be specified in InfoModeler.
They are not yet suited for the extended ORM concepts as introduced in 
\cite{Report:94:Halpin:ORMPoly}.
However, the modification of the here presented algorithms to cover these
extensions is not expected to be hard.

   \AddBib{asy}
   \BIBLIOGRAPHY{alpha}
\end{document}